\documentclass[10pt,conference]{IEEEtran}
\pdfoutput=1
\usepackage[utf8]{inputenc}
\usepackage[english]{babel}
\usepackage[T1]{fontenc}
\usepackage{amsmath}
\usepackage{amssymb}
\usepackage{bbm}
\usepackage{physics}
\usepackage{xcolor}
\usepackage{mathtools}
\usepackage{parskip}
\usepackage[margin=1in]{geometry}
\usepackage{algorithm}
\usepackage{algpseudocodex}
\usepackage{graphicx}
\usepackage{amsthm}
\usepackage{complexity}
\usepackage{multirow}
\usepackage{subcaption}
\usepackage{hyperref}

\newtheorem{lem}{Lemma}

\newtheorem{thm}{Theorem}

\theoremstyle{definition}

\title{Learning Gaussian Operations and the Matchgate Hierarchy}
\author{\IEEEauthorblockN{Josh Cudby}
\IEEEauthorblockA{DAMTP, Centre for Mathematical Sciences\\
University of Cambridge, Cambridge CB30WA, UK\\
Email: jjcc2@cam.ac.uk}
\and
\IEEEauthorblockN{Sergii Strelchuk}
\IEEEauthorblockA{DAMTP, Centre for Mathematical Sciences\\
University of Cambridge, Cambridge CB30WA, UK}
}

\date{April 2024}

\begin{document}

\maketitle
\begin{abstract}
    Learning an unknown quantum process is a central task for validation of the functioning of near-term devices.
    The task is generally hard, requiring exponentially many measurements if no prior assumptions are made on the process.
    However, an interesting feature of the classically-simulable Clifford group is that unknown Clifford operations may be efficiently determined from a black-box implementation.
    We extend this result to the important class of fermionic Gaussian operations. 
    These operations have received much attention due to their close links to fermionic linear optics.
    We then introduce an infinite family of unitary gates, called the Matchgate Hierarchy, with a similar structure to the Clifford Hierarchy.
    We show that the Clifford Hierarchy is contained within the Matchgate Hierarchy and how operations at any level of the hierarchy can be efficiently learned.
\end{abstract}
\begin{IEEEkeywords}
    Quantum Computing, Matchgates, Tomography, Learning.
\end{IEEEkeywords}

\section{Introduction}
Quantum process tomography is the task of determining an unknown quantum operator by applying it to certain known states.
In general, this task requires a number of measurements which grows exponentially with the number of qubits~\cite{chuang_prescription_1997, poyatos_complete_1997}.
This complexity can be improved if there is some prior knowledge about the process.
For example, restricting to the set of Pauli matrices reduces the query complexity to 1 for any number of qubits, for essentially the same reason as superdense coding~\cite{bennett_communication_1992} is possible.
A generalisation of learning Paulis was given by Low~\cite{low_learning_2009}, where they demonstrated how to learn any unknown Clifford Hierarchy operation. 

The Clifford Hierarchy $\mathcal{C}_k$ was introduced in~\cite{gottesman_quantum_1999} as an infinite family of unitary gates on any number of qubits which can be fault-tolerantly implemented using gate teleportation as a primitive.
Denoting the Pauli group by $\mathcal{C}_1$, the hierarchy is defined recursively for $k \geq 2$ by:
\begin{equation}
    \mathcal{C}_k \coloneqq \{ U : U \mathcal{C}_1 U^\dag \subseteq \mathcal{C}_{k-1} \} / U(1). 
\end{equation}
The Clifford group is the second level of the hierarchy, $\mathcal{C}_2$. 
It is well-known that circuits consisting only of Clifford gates and a single-qubit output measurement are classically simulable~\cite{gottesman_heisenberg_1998}.

The Clifford Hierarchy has received much attention due to its connections to fault-tolerance~\cite{zhou_methodology_2000,pllaha_-weyl-ing_2020,de_silva_efficient_2021}, magic state distillation~\cite{rengaswamy_unifying_2019} and learning~\cite{low_learning_2009}. 
The Clifford Hierarchy provides an elegant way to operationally describe the power of practically relevant classes of unitary operations: levels in the hierarchy represent a certain degree of hardness when implementing them.

Most pertinently to this work, it was shown in~\cite{low_learning_2009} that any Clifford gate $C$ on $n$ qubits may be learned (up to phase) using $\mathcal{O}(n)$ queries to black boxes for $C$ and $C^\dag$.
Moreover, they showed that any gate in $\mathcal{C}_k$ may similarly be learned, now requiring $\mathcal{O}((2n)^k)$ queries to the oracles.
The learning algorithms are based on the observation that knowledge of $U \sigma_p U^\dag$ for a set $\{\sigma_p\}$ which \emph{generates} $\mathcal{C}_1$ suffices to learn $U$.

Another class of important quantum operations is the set of \emph{(fermionic) Gaussian operations}, introduced in~\cite{valiant_quantum_2001} and studied in~\cite{terhal_classical_2002, jozsa_matchgates_2008} and more recently in~\cite{cudby_gaussian_2024,dias_classical_2023,reardon-smith_improved_2023}.
The authors of~\cite{terhal_classical_2002} described a deep relation between Gaussian operations and the physical theory of free fermions. It is then that authors first established that computational capabilities of unassisted fermionic linear optics can be precisely captured by these operations. Despite the fact that these operations are classically efficiently simulatable, their unique structure provides a complete characterisation of log-space bounded universal unitary quantum computation~\cite{jozsa_matchgate_2010}.

For a system of $n$ free fermions with creation operators $a_1^\dag,\,\ldots,\,a_n^\dag$ we define $2n$ \emph{Majorana operators}:
\begin{equation}
    \gamma_{2l-1} \coloneqq a_l + a_l^\dag, \qquad \gamma_{2l} \coloneqq -i(a_l - a_l^\dag),
\end{equation}
where $l \in [n] \coloneqq \{1,\,\ldots,\,n\}$.
The Majorana operators are Hermitian, square to the identity and anti-commute with each other.
For a set of indices $S \subseteq [2n]$, we denote by $\gamma_S$ the \emph{Majorana monomial} consisting of the product of the Majorana operators indexed by $S$ in increasing order.
When multiplying Majorana monomials together, we will use the notation $\gamma_S \gamma_{S'} = (\pm) \gamma_{S \Delta S'}$, where $\Delta$ denotes the symmetric difference of two sets and the possible factor of $-1$ comes from commuting the Majorana operators past each other.

We may choose the Jordan-Wigner representation~\cite{jordan_uber_1928} of the Majorana operators, giving
\begin{equation}
    \gamma_{2l - 1} = \left(\prod_{i = 1}^{l - 1} Z_i\right) X_l, \quad
    \gamma_{2l} = \left(\prod_{i = 1}^{l - 1} Z_i\right) Y_l,
\end{equation}
where $Z_i$ denotes the $n$-qubit operator which acts as the $1$-qubit Pauli $Z$ on the $i$th qubit and the identity elsewhere, and similarly for $X_i$ and $Y_i$.

Gaussian operations are the set of unitaries $M$ given by exponentials of \emph{(fermionic) quadratic Hamiltonians}:
\begin{align}\label{eq:quad_ham}
    H = i \sum_{\mu,\,\nu} h_{\mu \nu} \gamma_\mu \gamma_\nu, \qquad M = e^{iH}
\end{align}
with $h = (h_{\mu \nu})$ w.l.o.g. a real, anti-symmetric matrix. 

It can be shown~\cite{jozsa_matchgates_2008} that Gaussian operations preserve the linear span of $\gamma_\mu$ under conjugation:
\begin{equation}\label{eq:gaussian_op}
    M \gamma_\mu M^\dag = \sum_{\nu = 1}^{2n} Q_{\mu \nu} \gamma_\nu
\end{equation}
for a special orthogonal matrix $Q = (Q_{\mu \nu}) \in SO(2n)$, given by $Q = e^{4h}$.

An immediate consequence of~\eqref{eq:gaussian_op} is that:
\begin{equation}
    M_Q \gamma_S M_Q^\dag = \sum_{\abs{S'} = \abs{S}} \det(Q|_{S,\,S'}) \gamma_{S'},
\end{equation}
where the sum is over subsets of $[2n]$ of equal size to $S$, and $Q|_{S,\,S'}$ denotes the restriction of $Q$ to rows indexed by $S$ and columns indexed by $S'$.
In particular, the number of Majorana operators is preserved in each term.

Since the Majorana monomials form a basis of linear operators on $n$ qubits, $M$ is fully determined by~\eqref{eq:gaussian_op} up to overall phase.
We therefore label Gaussian operations by their corresponding orthogonal matrix, writing $M_Q$ for the operation defined by~\eqref{eq:gaussian_op}.

Gaussian operations may always be expressed as circuits of certain 2-qubit parity-preserving gates, known as matchgates. 
The matrices representing these gates have the following form:
\begin{equation}
    \begin{gathered}
    G(A, \, B) \coloneqq \begin{pmatrix}
        p & 0 & 0 & q \\
        0 & w & x & 0 \\
        0 & y & z & 0 \\
        r & 0 & 0 & s
    \end{pmatrix}, \\
    A = \begin{pmatrix}
        p & q \\
        r & s
    \end{pmatrix},
    \quad
    B = \begin{pmatrix}
        w & x \\
        y & z
    \end{pmatrix},
    \label{eq:matchgate}
    \end{gathered}
\end{equation}
where $\det(A) = \det(B)$.

In particular, a Gaussian operation $M$ on $n$ qubits is expressible as a circuit of at most $\mathcal{O}(n^3)$ matchgates $G(A,\,B)$~\cite{jozsa_matchgates_2008}. 
The set of Gaussian unitaries on $n$ modes forms a group, usually referred to as the \emph{matchgate group} and denoted by $\mathcal{G}_n$.

The set of (possibly mixed) \emph{Gaussian states} are given by the closure of Gibbs states of quadratic Hamiltonians:
\begin{equation}
    \rho_h = \frac{e^{-H}}{\Tr(e^{-H})}.
\end{equation}
For pure states, this coincides with the set of output states after applying some Gaussian operation to $\ket{0}$.

For any density matrix $\rho$, we define the \emph{correlation matrix} $\Gamma(\rho) = (\Gamma(\rho)_{jk})$ by
\begin{equation}
    \Gamma(\rho)_{jk} \coloneqq
\frac{i}{2} \langle [\gamma_j, \gamma_k] \rangle_\rho 
= \frac{i}{2}\Tr\left( [\gamma_j, \gamma_k] \rho \right).
\end{equation}
$\Gamma$ is easily seen to be real and anti-symmetric.
For a Gaussian state $\rho_h$, it can also be checked that $\Gamma(\rho_h) = \tan(2h)$. Since $\tan$ is injective for anti-symmetric matrices, we see that Gaussian states are uniquely defined by their correlation matrix.

In this work, we present an algorithm which (approximately) determines unknown Gaussian operations in polynomial query complexity. 
We then recursively define an infinite family of gates, mimicking the Clifford Hierarchy for the fermionic setting, which we call the Matchgate Hierarchy.
We then generalise our learning algorithm to show that elements from this hierarchy can also be efficiently determined.

Since the group of Gaussian operations is continuous, learning an unknown member of the group can only be done up to some finite precision.
Further, we are unable to distinguish physically equivalent unitaries which differ only by a global phase.
We, therefore, consider the two metrics used in~\cite{low_learning_2009}, which we refer to as the phase-sensitive and phase-insensitive distances respectively:
\begin{align}
    D^+(U_1,\,U_2) 
    &\coloneqq \frac{1}{\sqrt{2 d}} \norm{U_1 - U_2}_2 \nonumber\\
    &\equiv \sqrt{1 - \Re\left(d^{-1}\Tr(U_1^\dag U_2)\right)},\\
    D(U_1,\,U_2) 
    &\coloneqq \frac{1}{\sqrt{2 d^2}} \norm{U_1 \otimes U_1^* - U_2 \otimes U_2^*}_2 \nonumber\\
    &\equiv \sqrt{1 - \abs{d^{-1}\Tr(U_1^\dag U_2)}^2},
\end{align}
where $\norm{A}_2 = \sqrt{\Tr(A^\dag A)}$.
Although $D$ is not a true distance, we can check that: $D(U_1,\,U_2) = 0$ implies $U_1 = e^{i\phi}U_2$ for some phase $\phi$; $D$ is unitarily invariant; and $D$ satisfies the triangle inequality.
Since $D$ is phase-insensitive, it is a useful measure for the quality of our estimates for unknown Gaussian unitaries. 
\section{Learning Gaussian Operations}
Now suppose that $M = \exp(iH)$ is an unknown Gaussian operation.
If we have access to copies of the Gibbs state $\rho_h$, it is easy to learn $h$ and therefore $M$.
Since $\Gamma(\rho_h) = \tan(2h)$, we only need to compute estimates for the $n(2n-1)$ independent entries of $\Gamma(\rho_h)$ to obtain an estimate for $h$.
This process requires $\mathcal{O}(n(2n-1)/\epsilon^2)$ copies of the Gibbs state to obtain a precision of $\epsilon$.

Instead, suppose we only have access to the gates $M$ and $M^\dag$ as black boxes.

We consider a 3-step learning process.
\begin{enumerate}
    \item We perform a direct sampling process of each row of $Q$ in turn, obtaining a guess $\tilde{Q}$, which is identical to $Q$ up to additive error $\epsilon$ and entry-wise incorrect signs.
    \item We compute estimates for $\Gamma(\phi)_{jk}$ for certain pure Gaussian states $\ket{\phi} = M X_l\ket{0}$, refining our guess to $\overline{Q}$, which equals $Q$ up to overall signs for each row and additive error.
    \item We perform one final measurement to correct the signs for each row, leading to an additive error estimate of $Q$.
\end{enumerate}
The above protocol leads to the following result:

\begin{thm}\label{thm:m2_learn}
    Given oracle access to an unknown operation $M \in \mathcal{G}_n$ and its conjugate $M^\dag$, 
    an estimate $M'$ of $M$ satisfying $D(M,\,M') = \mathcal{O}(n^{3} \eta)$ can be determined w.h.p.~with query complexity $\tilde{\mathcal{O}}((n/\eta)^2)$, for any $\eta = \mathcal{O}(n^{-6})$.
\end{thm}
We prove the correctness and complexity of the protocol below.
Note that similar results on efficient tomography of Gaussian channels~\cite{oszmaniec_fermion_2022} and efficient learning of states prepared by circuits with only a small number of non-matchgate gates~\cite{mele_efficient_2024} are available.
The former uses a similar technique of learning the corresponding orthogonal matrix.
The latter considers a more general setting but learns only the Gaussian state rather than the whole unitary and relies on compressing the ``non-Gaussianity'' into a small state on which full tomography is performed. 

\subsection{Finding unsigned entries of the orthogonal matrix}
We first note that the set $\{ (\gamma_S M_Q \otimes I) \ket{\psi} : S \subseteq [2n] \}$ forms an orthonormal basis:
\begin{align*}
    \bra{\psi} (M_Q^\dag \gamma_{S'}^\dag& \otimes I)( \gamma_{S} M_Q \otimes I) \ket{\psi}\\
    & = \Tr (M_Q^\dag \gamma_{S'}^\dag \gamma_{S} M_Q)\nonumber\\
    &= \sum_{\abs{P} = \abs{S' \oplus S}} (\pm) \det(Q|_{S \oplus S',\,P}) \Tr(\gamma_P) \\
    &= \delta_{S,\,S'}
\end{align*}
since $\Tr(\gamma_S) = 0$ for any $S \neq \emptyset$.

Consider measuring the state $(M_Q \gamma_\mu \otimes I)  \ket{\psi}$ with respect to $\{( \gamma_S M_Q \otimes I) \ket{\psi} : S \subseteq [2n] \}$.
We have:
\begin{align}
    \mathbb{P}(S) 
    &= \abs{\bra{\psi} (M_Q^\dag \gamma_{S}^\dag \otimes I) (M_Q \gamma_\mu \otimes I)  \ket{\psi}}^2 \nonumber\\
    &= \abs{2^{-n}\Tr(M_Q^\dag \gamma_S^\dag  M_Q  \gamma_\mu)}^2 \nonumber\\
    &= \abs{2^{-n}\sum_{\abs{S'} = \abs{S}} \det(Q|_{S,\,S'}) \Tr(\gamma_{S'}\gamma_\mu) }^2 \nonumber \\
    &= \begin{cases}
        Q_{\mu \nu}^2 & S = \{ \nu \} \\
        0 & \text{otherwise}.
    \end{cases}
\end{align}
Suppose we take $K$ such measurements and report estimates for the squared entries $Q_{\mu \nu}^2$ according to the observed frequencies.

By the Hoeffding inequality and a union bound, we can learn each squared entry within additive error $\eta$ using $K = \mathcal{O}(\log(n) / \eta^2 )$ measurements w.h.p.

By taking the positive square roots, this is equivalent to learning the unsigned entries $\pm Q_{\mu \nu}$ up to additive error $\eta / 2$.

Repeating for each row requires a total of $\mathcal{O}(n \log(n) / \eta^2)$ measurements.

\subsection{Fixing signs up to row-wise sign errors}
We now have a matrix consisting of the unsigned entries of $Q$ with additive errors.
Let the matrix be $\tilde{Q} = \tilde{Q}_{\mu \nu}$.
Let the necessary sign corrections be $t_\mu s_{\mu \nu}$, where $t_\mu,\,s_{\mu \nu} \in \{\pm 1 \}$ and we have factored out whole-row sign corrections for each row.
Explicitly, 
\begin{align}
    \overline{Q}_{\mu \nu} &\coloneqq s_{\mu \nu} \tilde{Q}_{\mu \nu} + \eta_{\mu \nu} \\
    Q_{\mu \nu} &= t_\mu \overline{Q}_{\mu \nu} \nonumber\\
    &= t_\mu s_{\mu \nu} \tilde{Q}_{\mu \nu} + t_\mu \eta_{\mu \nu},
\end{align}
where $\abs{\eta_{\mu \nu}} < \eta$ and $s_{\mu,1} = +1$.

We will estimate certain 2-point correlation functions:
\begin{align}
    \Gamma^{(0)}_{jk} &\coloneqq \bra{0}M^\dag (i \gamma_j \gamma_k) M\ket{0}, \\
    \Gamma^{(l)}_{jk} &\coloneqq \bra{0}X_l M^\dag (i \gamma_j \gamma_k) M X_l\ket{0}.
\end{align}
We may compute an explicit expression for $\Gamma^{(0)}_{jk}$:
\begin{align}
    \Gamma^{(0)}_{jk} 
    &= i \sum_{a < b} \det(Q|_{\{a,\,b\},\,\{j,\,k\}}) \bra{0} \gamma_a \gamma_b \ket{0} \nonumber\\
    &= i \sum_{l = 1}^{n} \det(Q|_{\{2l-1,\,2l\},\,\{j,\,k\}}) \bra{0} i Z_l \ket{0} \nonumber\\
    &= - \sum_{l = 1}^{n} \det(Q|_{\{2l-1,\,2l\},\,\{j,\,k\}}).
\end{align}
Similarly,
\begin{align}
    \Gamma^{(l)}_{jk} 
    &= i \sum_{l' = 1}^{n} \det(Q|_{\{2l'-1,\,2l'\},\,\{jk\}}) \bra{0} i X_l Z_l' X_l \ket{0}\nonumber\\
    \begin{split}
        &=\det(Q|_{\{2l-1,\,2l\},\,\{jk\}}) \\
        &{}\ \ - \sum_{l' \neq l} \det(Q|_{\{2l'-1,\,2l'\},\,\{jk\}})
    \end{split}
\end{align}
We therefore see that:
\begin{equation}
     \frac{\Gamma^{(l)}_{jk} - \Gamma^{(0)}_{jk} }{2} = \det(Q|_{\{2l-1,\,2l\},\,\{jk\}})
\end{equation}
Setting $j = 1$ and expanding, we have:
\begin{align}
    C_{lk} \coloneqq \frac{\Gamma^{(l)}_{1k} - \Gamma^{(0)}_{1k} }{2} 
    &= Q_{2l-1,\,1}Q_{2l,\,k} - Q_{2l,\,1} Q_{2l-1,\,k} 
\end{align}
For clarity, we first consider the error-free case, where $\eta_{\mu \nu} = 0$.
Then our matrix of unsigned entries of $Q$ allows us to compute 4 guesses for the value of $C_{lk}$ (recalling that we fixed $s_{\mu,\,1} = +1$),
\begin{multline}
    {D}_{lk}^{\pm \pm} \coloneqq t_{2l-1}t_{2l} \Bigl( 
    {Q}_{2l-1,\,1} (\pm {Q}_{2l,\,k}) \\
    - {Q}_{2l,\,1} (\pm {Q}_{2l-1,\,k}) 
    \Bigr).
\end{multline}
If none of the 4 $D_{lk}^{\pm\pm}$ values are zero, then they are all distinct, and one of them exactly equals $C_{lk}$.
We would therefore fix the signs of ${Q}_{2l,\,k}$ and $Q_{2l-1,\,k}$ in terms of $t_{2l-1}t_{2l}$.
The set of matrices where any $D^{\pm\pm}_{lk} = 0$ has measure 0 and can safely be ignored.

In the errorful case, we have instead only approximations of $D_{lk}^{\pm\pm}$, given by:
\begin{multline}
    \tilde{D}_{lk}^{\pm \pm} \coloneqq t_{2l-1}t_{2l} \Bigl(  
    \tilde{Q}_{2l-1,\,1} (\pm \tilde{Q}_{2l,\,k}) \\
    - \tilde{Q}_{2l,\,1} (\pm \tilde{Q}_{2l-1,\,k}) 
    \Bigr).
\end{multline}
Furthermore, we only have an estimate $\tilde{C}_{lk}$ of $C_{lk}$ up to additive error $\epsilon$.
Obtaining these for each $C_{lk}$ w.h.p.~requires $\mathcal{O}(n^2 \log(n)/\epsilon^2)$ queries by the Hoeffding inequality and a union bound.

The true value $\tilde{D}_{lk}^{s_{2l-1,\,k} s_{2l,\,k}}$ satisfies:
\begin{align}
    \abs{\tilde{D}_{lk}^{s_{2l-1,\,k} s_{2l,\,k}} - \tilde{C}_{lk} } \leq 4\eta + \epsilon + \mathcal{O}(\eta^2),
\end{align}
using $\abs{Q_{\mu \nu}} \leq 1$.

Conversely, incorrect choices of sign are such that:
\begin{align}
    \abs{\tilde{D}_{lk}^{\pm \pm} - \tilde{C}_{lk} } \geq
    \abs{2F - (4\eta + \epsilon)},
\end{align}
where $F \in \{ \abs{Q_{2l,\,1}Q_{2l-1,\,k}},\, \abs{Q_{2l-1,\,1}Q_{2l,\,k}}, \, \abs{C_{lk}} \}$ for the 3 incorrect cases.

If $4\eta+\epsilon < \min \{ \abs{Q_{2l,\,1}Q_{2l-1,\,k}},\, \abs{Q_{2l-1,\,1}Q_{2l,\,k}}, \, \abs{C_{lk}}\}$, then we are guaranteed that the true choice of $\tilde{D}_{lk}^{\pm \pm}$ is the closest to $\tilde{C}_{lk}$.

In Appendix~\ref{app:sign_error_bounds}, we prove that
$\mathbb{P}(F < \kappa) = \mathcal{O}(\kappa^{1/3})$.
Taking a union bound over the $2n^2$ such calculations that we perform, we have that $\kappa = \mathcal{O}(n^{-6})$ suffices to ensure that all the quantities are larger than $\kappa$ with high probability.
Therefore, $\epsilon, \eta = \mathcal{O}(n^{-6})$ suffices to successfully learn the matrix $s$ w.h.p., fixing the signs up to overall row-wise errors.

\subsection{Fixing row-wise sign errors}
Assuming the above protocol succeeds, we have some matrix $\overline{Q}$, where $\overline{Q}_{\mu \nu} = t_\mu Q_{\mu \nu}$. 

For any set of rows $R$ where $t_\mu = -1$, $M_{\overline{Q}}$ is straightforwardly related to $M_Q$. 
Let $M_{R} \coloneqq M_Q \gamma_{R}$.
Then $M_R$ acts as:
\begin{align}
    M_{R} \gamma_\mu M_{R}^\dag 
    &= M_Q \gamma_R \gamma_\mu \gamma_R^\dag M_Q^\dag \nonumber \\
    &= (-1)^{\abs{R} - \mathbbm{1}_{\mu \in R}} M_Q \gamma_\mu M_Q^\dag \nonumber\\
    &= (-1)^{\abs{R}} \sum_{\nu} (-1)^{\mathbbm{1}_{\mu \in R}} Q_{\mu \nu}.
\end{align}
If $\abs{R}$ is even, then $M_{\overline{Q}} = M_R$.
If $\abs{R}$ is odd, then $M_{\overline{Q}} = M_{R^c}$, where $R^c$ is the complement of $R$ in $[2n]$.
Let $P$ be $R$ or $R^c$ corresponding to the 2 cases.

We can find the set $P$ by a single-shot distinguishing measurement.
Note that knowledge of $\overline{Q}$ is sufficient to write down an exact, poly-sized circuit of 2-qubit Gaussian operators which implements $M_Q^\dag$~\cite{jozsa_matchgates_2008}.
If we prepare $(M_{\overline{Q}}^\dag \otimes I)\ket{\psi}$ and measure w.r.t.~the basis $\{(\gamma_S M_Q^\dag \otimes I)\ket{\psi}\}$, we note that:
\begin{align}
    \mathbb{P}(S) 
    &= \abs{
    \bra{\psi}(M_Q \gamma_S M_{\overline{Q}}^\dag \otimes I )\ket{\psi}
    }^2 \nonumber\\
    &= \abs{2^{-n}\Tr ( M_Q \gamma_S \gamma_P^\dag M_Q^\dag)}^2 \nonumber\\
    &= \begin{cases}
        1 & S = P, \\
        0 & \text{otherwise}.
    \end{cases}
\end{align}
So the outcome of this measurement allows us to correct the remaining errors in the sign of each row of $\overline{Q}$, leaving us with the true matrix $Q$ up to additive error $\eta$.

In Appendix~\ref{app:estimate_error_bounds}, we prove that, w.h.p., this is sufficient to obtain an estimate $M_{Q'}$ of $M_Q$ satisfying:
\begin{align}
    D(M_Q,\,M_{Q'}) = \mathcal{O}(n^3 \eta).
\end{align}

\subsection{Complexity}
The total number of queries is easily seen to be $\mathcal{O}(n \log(n) /\eta^2 + n^2\log(n) / \epsilon^2 + 1)$.
Taking $\eta = \epsilon = \mathcal{O}(n^{-6})$, we obtain an estimate of $Q$ up to additive error $\eta$ using $\mathcal{O}((n/\eta)^2 \log(n))$ queries.

Thus the learning task may be completed in (high-degree) polynomial query complexity.
\section{Matchgate Hierarchy}
In the spirit of the Clifford Hierarchy, for any number of qubits $n$, we define an infinite family of gates $\mathcal{M}_k$ satisfying $\mathcal{M}_k \subseteq \mathcal{M}_{k+1}$ for each $k$, which we call the Matchgate Hierarchy.

We define the set $\Gamma_1 \coloneqq \{ \gamma_\mu : \mu \in [2n]\}$.
The $\gamma_\mu$ play the role of the Pauli generators $\{\sigma_{x_i},\,\sigma_{z_i}\}$.

We then recursively define the Matchgate Hierarchy by:
\begin{align}
    \mathcal{M}_1 &\coloneqq \{ M \in U(2^n) : M = \sum_\mu a_\mu \gamma_\mu ; a_\mu \in \mathbb{R} \} \\
    \mathcal{M}_k &\coloneqq \{ M \in U(2^n) : M \Gamma_1 M^\dag \subseteq \mathcal{M}_{k-1} \}.
\end{align}
The condition that $\mathcal{M}_1$ gates are unitary immediately gives $\sum_\mu a_\mu^2 = 1$.

We note that the set $\mathcal{M}_2$ consists of gates $M$ for which, for each $\mu \in [2n]$,
\begin{equation}
    M \gamma_\mu M^\dag = \sum_{\nu = 1}^{2n} Q_{\mu \nu} \gamma_\nu
\end{equation}
for some $Q \in \text{Mat}(2n,\,\mathbb{R})$.

We may evaluate the anti-commutator of $M\gamma_\mu M^\dag$ and $M\gamma_\nu M^\dag$ in 2 ways:
\begin{align}
    \{ M \gamma_\mu M^\dag,\, M \gamma_\nu M^\dag \} 
    &= M \{ \gamma_\mu,\,\gamma_\nu \} M^\dag \nonumber \\
    &= \delta_{\mu \nu} M M^\dag 
    = \delta_{\mu \nu} \\
    \{ M \gamma_\mu M^\dag,\, M \gamma_\nu M^\dag \}  
    &= \sum_{\sigma,\,\tau} Q_{\mu \sigma} Q_{\nu \tau} \{\gamma_\sigma,\,\gamma_\tau\} \nonumber \\
    &= \sum_{\sigma} Q_{\mu \sigma} Q^T_{\sigma \nu} 
    = (QQ^T)_{\mu \nu}
\end{align}
So we see that $(QQ^T)_{\mu \nu} = \delta_{\mu \nu}$, implying that $Q \in O(2n)$.
Such operations are sometimes referred to as \emph{extended Gaussian operations}.

As in the Clifford case, it is not straightforward to characterise the higher levels of the Hierarchy.
We do, however, have a link between the Matchgate Hierarchy and the Clifford Hierarchy:

\begin{lem}\label{lem:c_in_m}
    For any integer $k \geq 1$, $\mathcal{C}_k \subseteq \mathcal{M}_{k+1}$.
\end{lem}
\begin{proof}
    We proceed by induction.
    
    For the base case, note that $\mathcal{C}_1$ is the set of Pauli strings, which exactly coincides with the set of Majorana monomials $\{\gamma_S : S \subseteq [2n] \}$ up to factors of $\pm1$ or $\pm i$. Since $\gamma_S$ is easily seen to be a Gaussian operation, we have $\mathcal{C}_1 \subseteq \mathcal{M}_2$.

    For the inductive step, note that $\gamma_\mu \in \mathcal{C}_1$ for any $\mu \in [2n]$.
    Therefore, for any $C \in \mathcal{C}_k$, we have
    $C\gamma_\mu C^\dag \in \mathcal{C}_{k-1}$
    by definition of the Clifford hierarchy.
    By the inductive hypothesis, $\mathcal{C}_{k-1} \subseteq \mathcal{M}_{k}$. Since this is true for all $\mu \in [2n]$, we have $C \in \mathcal{M}_{k+1}$.
\end{proof}
A noteworthy implication of \autoref{lem:c_in_m} is that SWAP $\in \mathcal{M}_3$.
SWAP is a universality-enabling gate for matchgates and is in some sense the maximally non-matchgate 2-qubit gate. 
This is of potential interest for fault-tolerant implementation of quantum protocols based on a ``Matchgate + SWAP'' universal gate set.

\section{Learning Matchgate Hierarchy Operations}
The learning algorithm for Gaussian operations can be generalised to learning any operation from the $\mathcal{M}_k$ hierarchy.

We will require two Lemmas, given in~\cite{low_learning_2009}, which trivially extend to the current setting.
    \begin{lem}[Lemma 18,~\cite{low_learning_2009}]
        For unitaries $U_1,\,U_2$, if for all $\mu \in [2n]$
        \begin{align}
            D^+(U_1 \gamma_\mu U_1^\dag,\,U_2 \gamma_\mu U_2^\dag) \leq \delta
        \end{align}
        then for any $S \subset [2n]$
        \begin{align}
            D^+(U_1 \gamma_S U_1^\dag,\,U_2 \gamma_S U_2^\dag) \leq 2n\delta
        \end{align}
    \end{lem}
    \begin{proof}
        The proof follows by induction on the number of Majorana operators in $\gamma_S$ and the triangle inequality.
    \end{proof}
    \begin{lem}[Lemma 17,~\cite{low_learning_2009}]\label{lem:D_bound}
        For unitaries $U_1,\,U_2$, if for all $S \subset [2n]$
        \begin{align}
            D^+(U_1 \gamma_S U_1^\dag,\,U_2 \gamma_S U_2^\dag) \leq 2n\delta
        \end{align}
        then
        \begin{align}
            D(U_1,\,U_2) \leq 2n\delta
        \end{align}
    \end{lem}
    \begin{proof}
        The proof follows from the group of Paulis being a 1-design, and the fact that, under the Jordan-Wigner representation, each $\gamma_S$ is simply a Pauli string up to a factor of $i$.
    \end{proof}

\begin{thm}
    Given oracle access to an unknown operation $M \in \mathcal{M}_k$ for $k \geq 2$ and its conjugate $M^\dag$, 
    an estimate $M'$ of $M$ satisfying $D(M,\,M') = \mathcal{O}(n^{3} \eta)$ can be determined w.h.p.~with query complexity $\tilde{\mathcal{O}}(4^{3(k-2)} \cdot n^{2(k-2)} \cdot(n/\eta)^2)$ for any $\eta = \mathcal{O}(n^{-(6-k)}),\,o(1)$.
\end{thm}
\begin{proof}
    The proof is by induction, where the base case is given by~\autoref{thm:m2_learn} for $k = 2$.

    For the inductive step, suppose we have a learning algorithm for members of $\mathcal{M}_k$, which outputs estimates with phase-insensitive precision $\delta$.

    Let $M \in \mathcal{M}_{k+1}$, and for any $\mu \in [2n]$, let $M \gamma_\mu M^\dag = \alpha_\mu M_\mu$ for some $\alpha_\mu = \pm 1$.
    We may express $M_\mu$ in the Hermitian basis of Majorana monomials $\overline{\gamma}_{S} \coloneqq i^{f(\abs{S})} \gamma_S$, where the factors of $i$ are chosen to make $\overline{\gamma}_{S}^\dag = \overline{\gamma}_{S}$.
    Let $M_\mu = \sum_S c_S \overline{\gamma}_{S}$.
    Since $M_\mu$ is Hermitian, $c_S \in \mathbb{R}$ for each $S$, and we choose $\alpha_\mu$ so that $c_{\emptyset} \in \mathbb{R}_+$.
    
    The $M_\mu$ are elements of $\mathcal{M}_k$, so we may (phase-insensitively) learn them, obtaining estimates $M_{\mu}''$ satisfying $D(M_\mu,\, M_\mu'') \leq \delta$.
    We wish to find a phase-sensitive estimate by exploiting the Hermiticity of $M_\mu$.
    
    Let $M_\mu'' = \sum c_S'' \overline{\gamma}_{S}$.
    The vector of coefficients $c''$ may be written $c'' = \lambda c + \chi c^\perp$, for some $c^\perp$ satisfying $\langle c,\,c^\perp\rangle = 0$ for the usual complex inner product $\langle \cdot,\,\cdot \rangle.$

    It is easy to check that $D(M_\mu,\,M_\mu'') \leq \delta$ leads to $\abs{\langle c,\,c'' \rangle}^2 \geq 1 - \delta^2$, which in turn implies $\abs{\lambda}^2 \geq 1- \delta^2$.

    We wish to modify the phase of $M_\mu''$ in order to obtain a phase-sensitive estimate. 
    The optimal phase to multiply by is $-\arg(\lambda)$, but this is unknown.
    We approximate it by $\theta \coloneqq -\arg(\lambda c_S + \chi c_S^\perp)$ for any $S$.
    Some elementary trigonometry gives:
    \begin{equation}
        \abs{\arg(\lambda c_S + \chi c_S^\perp) - \arg(\lambda)} \leq \frac{\delta}{\sqrt{1 - \delta^2}}.
    \end{equation}
    Then we may compute the relevant real trace for $D^+(M_\mu,\, e^{-i\theta}M_\mu'')$:
    \begin{align}
        \Re\left(\Tr(M_\mu^\dag e^{-i \theta}M_\mu'') \right)
        &= d\Re\left( \abs{\lambda} e^{i(\arg(\lambda) - \theta)} \right) \nonumber\\
        &= d\abs{\lambda} \cos(\arg(\lambda) - \theta)\nonumber\\
        & \geq d\sqrt{1- \delta^2} \cdot \Big(1- \frac{\delta^2}{1 - \delta^2}\Big) \nonumber\\
        &= d\frac{1-2\delta^2}{\sqrt{1-\delta^2}}.
    \end{align}
    Substituting into the definition of $D^+$, we see that $D^+(M_\mu,\, e^{-i\theta}M_\mu'') \leq 2\delta$.  
    
    Now let $M''$ be the operator such that $M''\gamma_\mu M''^\dag = M_\mu''$.
    Also, let $M'$ be such that $M' \gamma_\mu M' = M_\mu$, i.e. all the phases $\alpha_\mu$ are set to $+1$.
    Similarly to the level-2 case, $M' = M \gamma_S$ for some set $S \subseteq [2n]$.
    By~\autoref{lem:D_bound}, we see:
    \begin{equation}\label{eq:M_distance}
        D(M',\, e^{-i\theta}M'') \leq 4n\delta.
    \end{equation}

    We may identify the phases $\alpha_\mu$ by preparing $(M^\dag M'' \otimes I)\ket{\psi}$ and measuring w.r.t.~the basis $\{ (\gamma_{S'}^\dag \otimes I)\ket{\psi} \}$.
    Then the probability of obtaining outcome $S$ is:
    \begin{align}
        \mathbb{P}(S) &= \abs{\bra{\psi}(\gamma_S \otimes I)(\gamma_S^\dag M'^\dag M'')\ket{\psi} }^2 \nonumber\\
        &= \abs{d^{-1}\Tr(M'^\dag M'')}^2\nonumber\\
        &\geq 1 - 16n^2 \delta^2,
    \end{align}
    using~\eqref{eq:M_distance}.

    The outcome of the measurement allows us to correct our estimate to $e^{-i\theta}M'' \gamma_S$, satisfying w.h.p.
    \begin{equation}
        D(M,\,e^{-i\theta}M'' \gamma_S) \leq 4n\delta.
    \end{equation}
    Obtaining the same precision as the $k = 2$ case thus requires setting $\eta' = \eta / (4n)^{k-2}$.
    
    For each additional level in the hierarchy, we must learn $2n$ operators from the previous level, requiring $4n$ queries to prepare them.

    Therefore, the query complexity at level $k$, $\mathcal{Q}(k)$, is given by:
    \begin{align}
        \mathcal{Q}(k) &=\mathcal{O}\left((4n)^{k-2} \cdot \left( \frac{n}{\eta/ (4n)^{k-2}} \right)^2 \right)\nonumber\\ 
        &= 4^{3(k-2)} \cdot n^{2(k-2)} \cdot (n/\eta)^2
    \end{align}
    as claimed.
\end{proof}

\section*{Acknowledgements}
JC thanks Wilfred Salmon for helpful discussions.
SS acknowledges support from
the Royal Society University Research Fellowship and
“Quantum simulation algorithms for quantum chromodynamics” grant (ST/W006251/1) and EPSRC Reliable and Robust Quantum Computing grant (EP/W032635/1).

\bibliographystyle{IEEEtran}
\bibliography{IEEEabrv,hierarchy.bib}

\appendices

\section{Error Bounds for Fixing Signs}
\label{app:sign_error_bounds}
We wish to obtain upper bounds on $\mathbb{P}(F < \kappa)$ for 
$ 
F \in \{ \abs{Q_{2l-1,\,1}Q_{2l,\,k}}, \, \abs{Q_{2l,\,1}\, Q_{2l-1,\,k}},\,
$ 
$ 
\abs{Q_{2l-1,\,1}\, Q_{2l,\,k} - Q_{2l,\,1}\, Q_{2l-1,\,k}} \},
$
for any $l \in [n],\,k \in [2n] \setminus \{1\}$, where the probability is over the Haar distribution on the orthogonal matrices.
We will fix $l = 1,\,k = 2$ and use a union bound to obtain the full result.

Since $Q$ is orthogonal, the first and second rows of $Q$ are random orthogonal unit vectors.
We may obtain them by performing Gram-Schmidt orthonormalisation on 2 random points on the unit sphere.

Let the two random points be given by the vectors $\mathbf{v_1},\,\mathbf{v_2}$.
We choose to parameterise them in polar coordinates by
\begin{equation}
    \mathbf{v_i} = \bigl( \sin(\theta^i_1),\, \cos(\theta^i_1) \cos(\theta^i_2),\,\ldots \bigr),
\end{equation}
where $\theta^i_1 \sim U[0,\,2\pi)$ and $\theta^i_2 \sim U[0,\,\pi)$, and all angles are distributed independently.
We are free to choose the ordering in the above parametrisation, i.e. which coordinate contains the $\sin(\theta^i_1)$ term.

We first consider $F_1 \coloneqq \abs{Q_{11}Q_{22}}$, with $F_2 \coloneqq \abs{Q_{21}Q_{12}}$ following similarly.

We have
\begin{align}
    \mathbb{P}(F_1 < \kappa ) 
    &\leq \mathbb{P}(\min(\abs{\cos(\theta^1_1),\,\cos(\theta^2_1)}) < \sqrt{\kappa}) \nonumber\\
    &\leq 2 \mathbb{P}(\abs{\cos(\theta^1_1)} < \sqrt{\kappa}) \nonumber\\
    &= \mathcal{O}(\sqrt{\kappa}),
\end{align}
to leading order.

We now consider $F_3 \coloneqq \abs{Q_{11} Q_{22} - Q_{21} Q_{12}} = \abs{\det(Q|_{\{1,\,2\},\,\{1,\,2\}})}$.
To begin, we consider the probability of this term being small if we did not perform the Gram-Schmidt procedure, denoting the corresponding quantity $\tilde{F_3}$.

Since the determinant of a 2-by-2 matrix is the area of the parallelogram defined by the rows of the matrix, we can express $\tilde{F_3}$ as
\begin{equation}
    \tilde{F_3} = \norm{\Pi (\mathbf{v_1})}  \norm{\Pi (\mathbf{v_2})} \sin(\phi),
\end{equation}
where $\Pi$ is the orthogonal projection onto the first 2 coordinates and $\phi$ is the angle between the projected vectors in the plane.
Noting that $\norm{\Pi (\mathbf{v_i})} = \abs{\cos(\theta^i_2)}$ and $\phi \sim U[0,\,2\pi)$, and using the shorthand $c(\cdot) = \cos(\cdot)$ and $s(\cdot) = \sin(\cdot)$, we may compute
\begin{align}
    \mathbb{P}(\tilde{F_3}& < \kappa) 
    = \mathbb{P}\left( \abs{c(\theta^1_2)c(\theta^i_2)s(\phi)} < \kappa \right) \nonumber\\
    \begin{split}
        &=\mathbb{P}\left( \abs*{c(\theta^1_2)c(\theta^i_2)s(\phi)} < \kappa \;\middle|\; \abs*{c(\theta^1_2)c(\theta^i_2)} < \zeta \right)\\
        & \qquad \ \cdot \mathbb{P}\left(\abs{c(\theta^1_2)c(\theta^i_2)} < \zeta\right) \\
        & \ + \mathbb{P}\left( \abs*{c(\theta^1_2)c(\theta^i_2)s(\phi)} < \kappa \;\middle|\; \abs*{c(\theta^1_2)c(\theta^i_2)} \geq \zeta \right)\\
        & \qquad \ \ \cdot \mathbb{P}\left(\abs{c(\theta^1_2)c(\theta^i_2)} \geq \zeta\right)
    \end{split}
    \nonumber\\
    &\leq \mathbb{P}\left(\abs{c(\theta^1_2)c(\theta^i_2)} < \zeta\right) + \mathbb{P}\left( \abs*{\zeta s(\phi)} < \kappa \right) \nonumber\\
    &= \mathcal{O}(\sqrt{\zeta}) + \mathcal{O}({\kappa}/{\zeta})
\end{align}
to leading order.

Optimising over $\zeta$, we obtain the bound $\mathbb{P}(\tilde{F_3} < \kappa) \leq \mathcal{O}(\kappa^{1/3})$.

We now consider the effect of orthonormalisation of $\mathbf{v_2}$ on the quantity $F_3$.
Let the resulting vector be $\mathbf{\tilde{v}_2} = \lambda \mathbf{{v}_2} + \mu \mathbf{{v}_1}$, where $\lambda \coloneqq (1- (\mathbf{{v}_1}\cdot \mathbf{{v}_2})^2)^{-0.5}$.

The corresponding determinant in the plane of the first two coordinates is simply $F_3 = \lambda \tilde{F_3}$, since the $\mathbf{v_1}$ part of $\mathbf{\tilde{v}_2}$ does not contribute.
Since $\lambda \in [1,\,\infty)$, we have
\begin{equation}
    \mathbb{P}({F_3} < \kappa) \leq \mathbb{P}(\tilde{F_3} < \kappa) = \mathcal{O}(\kappa^{1/3}).
\end{equation}

\section{Error Bounds for Estimates of Unknown Gaussian Operations}
\label{app:estimate_error_bounds}
Given a matrix $Q'$ which approximates $Q$ up to term-wise additive error $\eta$, we wish to bound the ``distance'' between the corresponding unitary operators $M_Q$ and $M_{Q'}$ as defined by Equations~\eqref{eq:quad_ham} and~\eqref{eq:gaussian_op}.

We first consider the error in the anti-symmetric matrix $h$ which generates $Q$ via $Q = \exp(4h)$.
Let $h' = h + \frac{1}{4}E$, where $E$ is the error term whose norm we will seek to bound.

We use the series expansion of the matrix logarithm:
\begin{multline}
    E = 4(h' - h) =
    \log(Q + \eta) - \log(Q) \\
    = \int_0^\infty dz \left\{
     \frac{I}{Q + zI} \eta \frac{I}{Q + zI} 
    \right\} + \mathcal{O}(\norm{\eta}_2^2).
\end{multline}
Therefore, to leading order,
\begin{multline}\label{eq:E_norm}
    \norm{E}_2 \leq \int_0^\infty dz \left\{
     \norm{\frac{I}{Q + zI}}_2 \norm{\eta}_2 \norm{\frac{I}{Q + zI}}_2 
    \right\}.
\end{multline}
We recall that for any matrix $A$,
\begin{align}
    \norm{A}_2 = \sqrt{\sum_i \sigma_i^2(A)} = \sqrt{\sum_i \sigma_i^{-2}(A^{-1})},
\end{align}
where the $\sigma_i(A)$ are the singular values of $A$.

Letting the eigenangles of $Q$ be 
$\{\theta_i\}_{i=1}^{2n}$, it is easy to check that
$\sigma_i(Q+zI) = \sqrt{1 + 2\cos(\theta_i)z + z^2}$.
Then Equation~\eqref{eq:E_norm} becomes
\begin{align}
    \norm{E}_2 &\leq \norm{\eta}_2 \int_0^{\infty} dz \left\{ 
        \sum_i \frac{1}{1 + 2\cos(\theta_i)z + z^2}
    \right\} \nonumber\\
    &= \norm{\eta}_2 \left( \sum_i \abs{\frac{\theta_i}{\sin(\theta_i)}}\right).
\end{align}
Since the eigenangles of a Haar-random orthogonal matrix in $2n$ dimensions are uniformly distributed (in complex conjugate pairs) on the unit disc~\cite{meckes_random_2019}, we may bound the probability of the term in brackets being overly large:
\begin{align}
    \mathbb{P}\left(\sum_i \abs{\frac{\theta_i}{\sin(\theta_i)}}\geq \kappa\right) 
    &\leq \mathbb{P}\left( \bigcup_i \abs{\frac{\theta_i}{\sin(\theta_i)}} \geq \frac{\kappa}{2n}\right) \nonumber\\
    &\leq 2n \mathbb{P}\left( \sum_i \abs{\frac{\theta}{\sin(\theta)}} \geq \frac{\kappa}{2n} \right)\nonumber\\
    &\leq 4n \mathbb{P}\left( \theta \in \big[\pi - \frac{2n\pi}{\kappa}, \pi\big)\right)\nonumber\\
    &= \frac{4n^2}{\kappa}.
\end{align}
Taking $\kappa = \Omega(n^2)$ then implies $\sum_i \abs{\frac{\theta_i}{\sin(\theta_i)}} < \kappa$ w.h.p.

We may also bound $\norm{\eta}_2$ using the entry-wise expression for the 2-norm:
\begin{align}
    \norm{\eta}_2 = \sqrt{\sum_{\mu \nu}\abs{\eta_{\mu \nu}}^2} \leq \sqrt{4n^2 \cdot \eta^2} = 2n\eta.
\end{align}
Therefore, $\norm{E}_2 = \mathcal{O}(n^3 \eta)$ w.h.p.

We are now prepared to compute the desired distance, 
$$D(M_Q,\,M_{Q'}) = \sqrt{1 - \abs{{d^{-1}\Tr(M_Q^\dag M_{Q'})}}^2}.$$
Using $\Tr(\exp(A+B))\leq \Tr(\exp(A)\exp(B))$ for Hermitian matrices $A,\,B$~\cite{bhatia_matrix_1996}, we find:
\begin{align}
    \Tr(M_Q^\dag M_{Q'}) &= \Tr(\exp(-iH)\exp(iH')) \nonumber\\
    &\geq \Tr(\exp(-iH + iH'))\nonumber \\
    &= \Tr(\exp(-\frac{1}{4}\sum_{\mu \nu}E_{\mu \nu}\gamma_\mu \gamma_\nu))\nonumber \\
    \begin{split}
        &= \Tr \Biggl(I - \frac{1}{4}\sum_{\mu \nu}E_{\mu \nu}\gamma_\mu \gamma_\nu)\\
        &+ \frac{1}{32} \sum_{\mu \nu \sigma \tau} E_{\mu \nu}E_{\sigma \tau} + \mathcal{O}(\norm{E}_2^3)) \Biggr)\nonumber
    \end{split}\\
    &= d\left(1 - \frac{1}{32}\sum_{\mu \nu}E_{\mu \nu}^2 + \mathcal{O}(\norm{E}_2^4)\right),
\end{align}
recalling that $\Tr(\gamma_S) = 0$ for $S \neq \emptyset$.
Finally, to leading order we have:
\begin{align}
    D(M_Q,\,M_{Q'}) 
    &= \sqrt{1- \abs{1 - \frac{1}{32}\sum_{\mu \nu}E_{\mu \nu}^2}^2} \nonumber \\
    &\leq \sqrt{\frac{1}{16}\abs{\sum_{\mu \nu}E_{\mu\nu}^2}} \nonumber\\
    &\leq \frac{1}{4}\sqrt{\sum_{\mu \nu}\abs{E_{\mu \nu}}^2 } \nonumber\\
    &= \frac{1}{4}\norm{E}_2 = \mathcal{O}(n^3 \eta).
\end{align}

\end{document}